\documentclass[reqno,11pt]{amsart}
\usepackage{color}
\IfFileExists{mymtpro2.sty}{%
  \usepackage[subscriptcorrection]{mymtpro2}
}{}

\usepackage{a4,amssymb}
\usepackage{graphicx}

\newtheorem{theorem}{Theorem}

\newtheorem{thm}{Theorem}[section]
\newtheorem{lemma}{Lemma}[section]

\newtheorem{pr}{Proposition}[section]
\newtheorem{defn}{Definition}[section]

\newcommand{\bel}{\begin{equation} \label}
\newcommand{\ee}{\end{equation}}

\theoremstyle{remark}
\newtheorem{remark}[theorem]{Remark}
\newtheorem{myremarks}[theorem]{Remarks}

    {\end{nummer}\end{myremarks}}

\newcounter{numcount}
\newcommand{\labelnummer}{\mbox{\normalfont (\roman{numcount})}}%

\makeatletter

\newenvironment{nummer}%
  {\let\curlabelspeicher\@currentlabel%
    \begin{list}{\labelnummer}%
      {\usecounter{numcount}\leftmargin0pt%
        \topsep0.5ex\partopsep2ex\parsep0pt\itemsep0ex\@plus1\p@%
        \labelwidth2.5em\itemindent3.5em\labelsep1em%
      }%
    \let\saveitem\item%
    \def\item{\saveitem%
      \def\@currentlabel{{\upshape\curlabelspeicher}$\,$\labelnummer}}%
    \let\savelabel\label%
    \def\label##1{\savelabel{##1}%
      \@bsphack%
        \ifmmode\else%
          \protected@write\@auxout{}%
          {\string\newlabel{##1item}{{\labelnummer}{\thepage}}}%
        \fi%
      \@esphack%
    }%
  }{\end{list}}%

\renewcommand{\appendix}{\def\thesection{\textsc{Appendix}}}

 \let\leq\le
 \let\geq\ge

\DeclareMathOperator{\tr}{tr\kern1pt}

\newcommand{\Ran}{\mathop\mathrm{Ran}\nolimits}

\newcommand{\supp}{\mathop\mathrm{supp}\nolimits}

\newcommand{\N}{{\mathbb N}}

\makeatletter

\newif\ifper\pertrue
\def\per{.}

\def\bti{\@ifnextchar[\bbti\bbbti}
\def\bbti[#1]#2{#2, #1.}
\def\bbbti#1{#1.}

\def\z{\@ifnextchar[\zz\zzz}
\def\zz[#1]#2#3#4#5{\perfalse\emph{#2} \textbf{#3}, #4 (#5) [#1]}
\def\zzz#1#2#3#4{\emph{#1} \textbf{#2}, #3 (#4)\ifper\per\fi\pertrue}

\def\pub{\@ifstar\pubstar\pubnostar}
\def\pubnostar{\@ifnextchar[\@@pubnostar\@pubnostar}
\def\@@pubnostar[#1]#2#3#4{#2, #3, #4, #1\ifper\per\fi\pertrue}
\def\@pubnostar#1#2#3{#1, #2, #3\ifper\per\fi\pertrue}
\def\pubstar[#1]#2#3#4{\perfalse #2, #3, #4 [#1]\pertrue}

\makeatother

\newcommand{\beq}{\begin{equation}}
\newcommand{\eeq}{\end{equation}}
\newcommand{\ba}{\begin{array}}
\newcommand{\ea}{\end{array}}
\newcommand{\bea}{\begin{eqnarray}}
\newcommand{\eea}{\end{eqnarray}}
\newcommand{\beas}{\begin{eqnarray*}}
\newcommand{\eeas}{\end{eqnarray*}}

{

\newcommand{\R}{\mathbb{R}}

\newcommand{\T}{\mathbb{T}}

\newcommand{\Z}{\mathbb{Z}}

\newcommand{\C}{\mathbb{C}}

\def\P{I\kern-.30em{P}}
\def\E{I\kern-.30em{E}}
\renewcommand{\E}{\mathbb{E}\mkern2mu}
\renewcommand{\P}{\mathbb{P}}

\begin{document}

\title[Decorrelation estimates for non rank one perturbations]{Decorrelation estimates for random Schr\"odinger operators with non rank one perturbations}

\author[P.\ D.\ Hislop]{Peter D.\ Hislop}
\address{Department of Mathematics,
    University of Kentucky,
    Lexington, Kentucky  40506-0027, USA}
\email{peter.hislop@uky.edu}

\author[M.\ Krishna]{M.\ Krishna}
\address{Ashoka University, Rai,
Haryana, India}
\email{krishna.maddaly@ashoka.edu.in}


\author[C.\ Shirley]{C.\ Shirley}
\address{Universit\'e Libre de Bruxelles, Belgique}
\email{cshirley@ulb.ac.be}

\thanks{PDH was partially supported by NSF through grant DMS-1103104 during the time some of this work was done.
MK was partially supported by IMSc Project 12-R\&D-IMS-5.01-0106. The authors thank F.\ Klopp for discussions on eigenvalue statistics and decorrelation estimates.}


\begin{abstract}
   We prove decorrelation estimates for generalized lattice Anderson models on $\Z^d$ constructed with finite-rank perturbations in the spirit of Klopp \cite{klopp}. These are applied to prove that the local eigenvalue statistics $\xi^\omega_{E}$ and
   $\xi^\omega_{E^\prime}$, associated with two energies $E$ and $E'$ in the localization region and satisfying $|E - E'| > 4d$, are independent. That is, if $I,J$ are two bounded intervals, the random variables $\xi^\omega_{E}(I)$ and $\xi^\omega_{E'}(J)$, are independent and distributed according to a compound Poisson distribution whose L\'evy measure has finite support. We also prove that the extended Minami estimate implies that the eigenvalues in the localization region have multiplicity at most the rank of the perturbation. The method of proof contains new ingredients that simplify the proof of the rank one case \cite{klopp,shirley,trinh}, extends to models for which the eigenvalues are degenerate, and applies to models for which the potential is not sign definite \cite{tautenhahn-veselic1} in dimensions $d \geq 1$.
 \end{abstract}

\maketitle \thispagestyle{empty}



\tableofcontents




\section{Statement of the problem and results}\label{sec:introduction}
\setcounter{equation}{0}

We consider random Schr\"odinger operators $H^\omega = {\mathcal L} + V_\omega$ on the lattice Hilbert space $\ell^2 (\Z^d)$ (or, for matrix-valued potentials, on $\ell^2 (\Z^d) \otimes \C^{m_k}$),
and prove that certain natural random variables
associated with the local eigenvalue statistics around two distinct energies $E$ and $E^\prime$, in the region of complete localization $\Sigma_{\rm CL}$ and with $| E - E^\prime | > 4d$, are independent. From previous work \cite{hislop-krishna1}, these random variables
distributed according to a compound Poisson distribution.
The operator ${\mathcal L}$ is the discrete Laplacian on $\Z^d$, although this can be generalized.
For these lattice models, the random potential $V_\omega$ has the form
\beq\label{eq:potential1}
(V_\omega f)(j) = \sum_{i \in {\mathcal J}} \omega_i  (P_i f)(j),
\eeq
where $\{ P_i \}_{i \in {\mathcal{J}}}$ is a family of finite-rank projections with the same rank $m_k \geq 1$, the set $\mathcal{J}$ is a sublattice of $\Z^d$, and $\sum_{i \in {\mathcal{J}}} P_i = I$. We assume that $P_i = U_i P_0 U_i^{-1}$, for all $i \in \mathcal{J}$,
where $U_i$ is the unitary implementation of the translation group $(U_i f)(k) = f(k+i)$, for $i, k \in \Z^d$.
The coefficients $\{ \omega_i \}$ are a family of independent, identically distributed (iid) random variables with a bounded density of compact support on a product probability space $\Omega$ with probability measure $\P$.
It follows from the conditions above that the family of random Schr\"odinger operators $H^\omega$ is ergodic with respect to the translations generated by the sublattice $\mathcal{J}$.

One example on the lattice is the polymer model. For this model, the projector $P_i = \chi_{\Lambda_k(i)}$ is the characteristic function
on the cube $\Lambda_k(i)$ of side length $2k \in \N$ centered at $i \in \Z^d$. The rank of $P_i$ is $(2k+1)^d$
and the set $\mathcal{J}$ is chosen so that $\cup_{i \in \mathcal{J}} {\Lambda_k(i)} = \Z^d$.
Another example is a matrix-valued model for which $P_i$, $i \in \Z^d$, projects onto the $m_k$-dimensional subspace $C^{m_k}$,
 and $\mathcal{J} = \Z^d$. The corresponding Schr\"odinger operator is
\beq\label{eq:model1}
H^\omega = {\mathcal L} +  \sum_{i \in {\mathcal{J}}} \omega_i  P_i ,
\eeq
where ${\mathcal L}$ is the discrete lattice Laplacian $\Delta$ on $\ell^2 (\Z^d)$,
or $\Delta \otimes I$ on $\ell^2(\Z^d) \otimes \C^{m_k}$ or, more generally, $\Delta \otimes A$, where $A$ is a hermitian positive-definite $m_k \times m_k$ matrix), respectively.
In the following, we denote by $H_{\omega, \ell}$ (or simply as
$H_\ell$ omitting the $\omega$) the matrices $\chi_{\Lambda_\ell}H^\omega\chi_{\Lambda_\ell}$ and similarly $H_{\omega,L}, H_L$ by replacing $\ell$ with $L$, for positive integers $\ell$ and $L$.

A lot is known about the eigenvalue statistics for random Schr\"odinger operators on $\ell^2 (\R^d)$.
When the projectors $P_i$ are rank one projectors, the local eigenvalue statistics in the localization regime has been proved to be given by a Poisson process by Minami \cite{minami1} (see also Molchanov \cite{molchanov} for a model on $\R$ and Germinet-Klopp \cite{germinet-klopp} for a comprehensive discussion and additional results). For the non rank one case, Tautenhahn and Veseli\'c \cite{tautenhahn-veselic1} proved a Minami estimate for certain models that may be described as weak perturbations of the rank one case. The general non finite rank case was studied by the first two authors in \cite{hislop-krishna1} who proved that, roughly speaking, the local eigenvalue statistics in the localization regime are compound Poisson point processes. This result also holds for random Schr\"odinger operators on $\R^d$.

In this paper, we further refine these results for lattice models with non rank one projections and prove, roughly speaking, that the processes associated with two energies are independent. The method applies to Schr\"odinger operators with non simple eigenvalues, see the discussion at the end of section \ref{sec:multiplicity1}. Klopp \cite{klopp} proved decorrelation estimates for rank one lattice models in any dimension. He applied them to show that the local eigenvalue point processes at distinct energies converge to independent Poisson processes (in dimensions $d \>1$ the energies need to be far apart as is the case for the models studied here).
Shirley \cite{shirley} extended the family of one-dimensional lattice models for which the decorrelation estimate may be proved to include alloy-type models with correlated random variables, hopping models, and certain one-dimensional quantum graphs.

One of the advantages of the methods employed in this paper is that monotonicity is no longer needed. Consequently, we can treat the almost rank one models for which the potential is not sign definite. A class of such models was considered by Tautenhahn and Veseli\`c \cite{tautenhahn-veselic1} who proved a Minami estimate.

\subsection{Asymptotic independence and decorrelation estimates}\label{subsec:decorrelation1}

The main result is the asymptotic independence of random variables associated with the local eigenvalue statistics
centered at two distinct energies $E$ and $E^\prime$ satisfying $|E-E^\prime | > 4d$.

 We note that in one-dimension there are stronger results and the condition $|E-E^\prime | > 4d$ is not needed. Our results are inspired by the work of Klopp \cite{klopp} for the Anderson models on $\Z^d$ and of Shirley \cite{shirley} for related models on $\Z^d$.
The condition $|E - E^\prime | > 4d$ requires that the two energies be fairly far apart. For example, if $\omega_0 \in [-K, K]$ so that the deterministic spectrum $\Sigma = [ -2d - K, 2d + K]$, the region of complete localization $\Sigma_{\rm CL}$ is near the band edges $\pm ( 2d + K)$. In this case, one can consider $E$ and $E^\prime$ near each of the band edges.
Our main result on asymptotic independence is the following theorem.

\begin{thm}\label{thm:decorrelation-lattice1}
Let $E,E' \in \Sigma_{\rm CL}$ be two distinct energies with $|E - E'| > 4d$. Let $\xi_{\omega, E}$, respectively, $\xi_{\omega, E'}$,  be a limit point of the local eigenvalue statistics centered at $E$, respectively, at $E'$. Then these two processes are independent. That is, for any bounded intervals $I,J \in \mathcal{B}(\R)$, the random variables $\xi_{\omega, E}(I)$ and $\xi_{\omega, E'}(J)$ are independent random variables distributed according to a compound Poisson process.
\end{thm}

We refer to \cite{germinet-klopp} for a description of the region of complete localization $\Sigma_{\rm CL}$. For information on L\'evy processes, we refer to the books by Applebaum \cite{applebaum} and by Bertoin \cite{bertoin}.
Theorem \ref{thm:decorrelation-lattice1} follows (see section \ref{sec:proof1}) from the following decorrelation estimate. We assume that $L > 0$ is a positive integer, and that $\ell := [ L^\alpha]$ is the greatest integer less than $L^\alpha$ for an exponent $0 < \alpha < 1$. For polymer type models, we assume that $m_k$ divides $L$ and $\ell$.

\begin{pr}\label{prop:decorrelation-prop1}
We choose any $0 < \alpha < 1$ and $\beta > \frac{5}{2}$,
and length scales $L$ and $\ell := [L^\alpha]$ as described above. For a pair of energies $E, E^\prime \in \Sigma_{\rm CL}$, the region of complete localization, with $\Delta E := |E - E^\prime| > 4d$, and bounded intervals $I,J \subset \R$, we define $I_L(E) := L^{-d}I + E$ and $J_L (E^\prime) := L^{-d} J + E^\prime$ as two scaled energy intervals centered at $E$ and $E^\prime$, respectively. There exists $L_0 = L_0 (\alpha, \beta, d)> 0$ and a constant $C_0 = C_0 (L_0)$ such that for all $L > L_0$ we have
\beq\label{eq:prob1}
\P \{ ({\rm Tr} E_{H_{\omega,\ell}} (I_L (E)) \geq 1 )  \cap  ({\rm Tr} E_{H_{\omega,\ell}} (J_L (E^\prime)) \geq 1 ) \} \leq C_0 \frac{K \| \rho \|_\infty^2 m_k^{2 m_k} }{\Delta E - 4d}  \frac{ (C \log L)^{(1+\beta) d}}{L^{(2-\alpha) d}} . 
\eeq
\end{pr}

The extended Minami estimate \cite{hislop-krishna1} (see section \ref{subsec:reduction1}) implies that we only need to estimate the probability that there is a small number of eigenvalues in each interval:
\beq\label{eq:prob2}
\P \{ ({\rm Tr} E_{H_{\omega,\ell}} (I_L (E)) \leq m_k )  \cap  ({\rm Tr} E_{H_{\omega,\ell}} (J_L(E^\prime)) \leq m_k ) \}
\eeq
In fact, we consider the more general estimate:
\beq\label{eq:prob3}
\P \{ ({\rm Tr} E_{H_{\omega,L}} (I_L(E)) = k_1 )  \cap  ({\rm Tr} E_{H_{\omega,L}} (J_L(E^\prime)) = k_2 ) \},
\eeq
where $k_1, k_2 \leq m_k$ are positive integers independent of $L$.

We allow that there may be several eigenvalues in $I_L(E)$ and $J_L (E^\prime)$ with nontrivial multiplicities.
To deal with this, we introduce the mean trace of the eigenvalues $E_j(\omega)$ of $H_{\omega, \ell}$ in the interval $I_L(E)$:
\beq\label{eq:defn-weight-ave1}
\mathcal{T}^\ell (E, k_1, \omega) :=
\frac{Tr\big(H_{\omega,\ell} E_{H_{\omega,\ell}}(I_L(E))\big)}{Tr\big(E_{H_{\omega,\ell}}(I_L(E))\big)} = \frac{1}{k_1} \sum_{j=1}^{k_1} E_j^\ell(\omega),
\eeq
where $k_1 := Tr\big(E_{H_{\omega,\ell}}(I_L(E))\big)$ is the number of eigenvalues, including multiplicity,
of $H_{\omega, \ell}$ in $I_L(E)$.
Similarly, we define
\beq\label{eq:defn-weight-ave2}
\mathcal{T}_\ell(E^\prime, k_2, \omega) := \frac{Tr\big(H_{\omega,\ell} E_{H_{\omega,\ell}}(J_L(E^\prime))\big)}{Tr\big(E_{H_{\omega,\ell}}(J_L(E^\prime))\big)}
= \frac{1}{k_2} \sum_{j=1}^{k_2} E_j^\ell(\omega),
\eeq
where $k_2 := Tr\big(E_{H_{\omega,\ell}}(J_L(E^\prime))\big)$. We will show in section \ref{sec:sum-estimates1} that these weighted sums behave like effective
eigenvalues in each scaled interval $I_L(E)$ and $J_L(E^\prime)$, respectively.

As another application of the extended Minami estimate, we prove that the multiplicity of eigenvalues in $\Sigma_{\rm CL}$ is at most the multiplicity of the perturbations $m_k$ in dimensions $d \geq 1$. The proof of this fact follows the argument of Klein and Molchanov \cite{klein-molchanov}. For $d=1$, Shirley \cite{shirley} proved that the usual Minami estimate holds for the dimer model ($m_k = 2$) so the eigenvalues are almost surely simple.


\subsection{Contents}\label{subsec:contents}

We present properties of the average of eigenvalues in section \ref{sec:sum-estimates1}, including gradients estimates.The proof of the main technical result, Proposition \ref{prop:decorrelation-prop1}, is presented in section \ref{sec:proof-prop1}. The proof of asymptotic independence is given in section \ref{sec:proof1}. We show in section \ref{sec:multiplicity1} that the argument of Klein-Molchanov \cite{klein-molchanov} applies to higher rank perturbations and implies that the multiplicity of eigenvalues in $\Sigma_{\rm CL}$ is at most $m_k$, the uniform rank of the perturbations. In section \ref{sec:alloy-type1}, we prove that the decorrelation estimates, and therefore
asymptotic independence of local eigenvalue processes, hold for non sign definite models studied by Tautenhahn and Veseli\`c \cite{tautenhahn-veselic1}. This paper replaces the manuscript \cite{hislop-krishna2} by the first two authors, completing and improving the arguments, and extending the results. 


\section{Estimates on weighted sums of eigenvalues}\label{sec:sum-estimates1}
\setcounter{equation}{0}

In this section, we present some technical results on weighted sums of eigenvalues of $H_{\omega, \ell}$
defined in \eqref{eq:defn-weight-ave1}-\eqref{eq:defn-weight-ave2}.
These are used in section \ref{sec:proof1} to prove the main technical result \eqref{eq:prob1}.
We recall that $\ell = [L^\alpha]$, for $0 < \alpha < 1$.


\subsection{Properties of the weighted trace}\label{subsec:weight-trace1}

When the total number of eigenvalues of $H_{\omega,\ell}$ in $I_L(E) := L^{-d}I + E$ is $k_1$, we have
\beq\label{defn:trace3}
\mathcal{T}(\omega) := \mathcal{T}_\ell(E,k_1) := \mathcal{T}_\ell (E, k_1, \omega)
= \frac{1}{k_1} \sum_{j=1}^{k_1} E_j(\omega),
\eeq
for eigenvalues $E_j^\ell(\omega) \in {I}_L (E)$.
Properties (1)-(3) below are valid
for the similar expression obtained by replacing $k_1$ with $k_2$, the interval $I$ with $J$, and the energy $E$ with $E^\prime$. We will write
$$
\mathcal{T}^\prime(\omega) := \mathcal{T}_\ell(E^\prime, k_2) := \mathcal{T}(E^\prime, k_2, \omega).
$$

We write $P_E(\omega)$ for the spectral projection $E_{H_{\omega, \ell}}(I_L(E))$ onto the eigenspace of $H_{\omega, \ell}$ corresponding to the eigenvalues $E_m^\ell(\omega)$ in $I_L(E)$. Let $\gamma_E$ be a simple closed contour containing only these eigenvalues of $H_{\omega, \ell}$ with a counter-clockwise orientation. Since the mean of the eigenvalues may be expressed as
$$
\mathcal{T}(\omega) = \frac{1}{k_1} {\rm Tr} H_{\omega,\ell} P_E (\omega),
$$
and the projection has the representation
$$
P_E (\omega) = \frac{-1}{2 \pi i} \int_{\gamma_E} R(z) ~dz, ~~~ R(z) := (H_{\omega, \ell} - z)^{-1}.
$$

The Hamiltonian $H_{\omega, \ell}$ is analytic in the variables $\omega_j$. The projection $P_E(\omega)$ is also analytic in $\omega$ provided the contour $\gamma_E$ remains isolated from the other eigenvalues of $H_{\omega, \ell}$. We work on that part of the probability space for which the total multiplicity of the eigenspace ${\rm Ran}~ P_E(\omega) = k_1$, that is, on the set $\Omega(k_1) := \{ \omega ~|~ {\rm Tr} P_E(\omega) = k_1 \}$.
If we place a security zone around $I_L(E)$ of width $L^{-d}$ then the probability that $H_{\omega, \ell}$ has no eigenvalues in this zone is larger than $1 - (\ell / L)^d$ by the Wegner estimate.
On this set, it follows that
\beq\label{eq:tr-deriv1}
\frac{\partial \mathcal{T}(\omega) }{\partial \omega_j} = \frac{1}{2 \pi i k_1} \int_{\gamma_E} {\rm Tr} \{ R(z) P_j R(z) \} ~z dz ,
\eeq
where $P_j$ is the finite-rank projector associated with site $j$ or block $j$, depending on the model.
Evaluating the contour integral, we find that
\beq\label{eq:tr-deriv2}
\frac{\partial \mathcal{T}(\omega) }{\partial \omega_j} = \frac{1}{k_1} {\rm Tr} \{ P_E (\omega) P_j  \} .
\eeq

Formula \eqref{eq:tr-deriv2} shows that the eigenvalue average behaves like an effective eigenvalue in the following sense:
\begin{enumerate}
\item $\mathcal{T}_\ell (E, k_1, \omega) \in {I}_L (E)$,
 so the average of the eigenvalue cluster in ${I}_L(E)$ belongs to ${I}_L (E)$.


\item The $\omega_j$-derivative of $\mathcal{T}(\omega)$ is nonnegative as follows directly from \eqref{eq:tr-deriv2}.

\item Since $\sum_{j \in \mathcal{J} \cap \Lambda_\ell} P_j = I_{\Lambda_\ell}$, it follows from this and \eqref{eq:tr-deriv2}
that the $\omega$-gradient of the weighted trace is normalized:
$\| \nabla_\omega \mathcal{T}(\omega) \|_{\ell^1} = 1$.




\end{enumerate}

\begin{remark}\label{remark:t-separation}
It follows from property (1) above and the fact that the intervals $I_L(E)$ and $J_L(E)$ are $\mathcal{O}(L^{-d})$, that if $|E - E^\prime| > 4d$,
then $|\mathcal{T}(\omega) - \mathcal{T}(\omega^\prime)| > 4d - cL^{-d}$, for some $c > 0$.
We will use this result below.
\end{remark}



\subsection{Variational formulae}\label{subsec:variation1}

We can estimate the variation of the mean trace with respect to the random variables as follows. The $\omega$-directional derivative is
 \bea\label{eq:variation2}
 \omega \cdot \nabla_\omega (\mathcal{T}(\omega) - \mathcal{T}^\prime (\omega) ) &=& \sum_{j \in \mathcal{J} \cap \Lambda_\ell} \omega_j \left\{ \frac{1}{k_1} {\rm Tr} P_E(\omega) P_j - \frac{1}{k_2} {\rm Tr} P_{E^\prime} (\omega) P_j \right\} \nonumber \\
 &=& \left\{ \frac{1}{k_1} {\rm Tr} P_E(\omega) V^\omega_{\Lambda_\ell} - \frac{1}{k_2} {\rm Tr} P_{E^\prime} (\omega) V^\omega_{\Lambda_\ell} \right\} \nonumber \\
 &=&  \left\{ \frac{1}{k_1} {\rm Tr} P_E(\omega) H_{\omega, \ell} - \frac{1}{k_2} {\rm Tr} P_{E^\prime} (\omega) H_{\omega, \ell} \right. \nonumber \\
  & &  \left. - \frac{1}{k_1} {\rm Tr} P_E(\omega) \mathcal{L} + \frac{1}{k_2} {\rm Tr} P_{E^\prime} (\omega) \mathcal{L}  \right\}.
  \eea


 The absolute value of each trace involving the Laplacian $\mathcal{L}$ in \eqref{eq:variation2} may be bounded above by $2 d$. If we assume that
$$
| \mathcal{T}(\omega) - \mathcal{T}^\prime(\omega) | \geq \Delta E
$$
then we obtain from \eqref{eq:variation2},
\bea\label{eq:variation3}
\Delta E - 4d & \leq & | \mathcal{T}(\omega) - \mathcal{T}^\prime(\omega) | - 4d \nonumber \\
 & \leq & |\omega \cdot \nabla_\omega (\mathcal{T}(\omega) - \mathcal{T}^\prime (\omega) )|.
\eea
As the number of components of $\omega$ is bounded by $\ell^d$ and $|\omega_j | \leq K$, it follows by Cauchy-Schwartz inequality that
\beq\label{eq:variation4}
\| \nabla_\omega (\mathcal{T}(\omega) - \mathcal{T}^\prime (\omega) )\|_2 \geq \frac{\Delta E - 4d}{K} \frac{1}{(2\ell + 1)^{d/2}}.
\eeq
We also obtain an $\ell^1$ lower bound:
\beq\label{eq:variation5}
\| \nabla_\omega (\mathcal{T}(\omega) - \mathcal{T}^\prime (\omega) )\|_1 \geq \frac{\Delta E - 4d}{K} .
\eeq


\subsection{Dependence of the weighted eigenvalue averages on the random variables}\label{subsec:dependence-rv1}

Suppose that $H_{\omega, \ell}$ has $k_1$ eigenvalues (including multiplicities) $\{ E_m^\ell (\omega); m=1, \ldots, k_1 \}$ in an interval $I \subset \R$.  The corresponding weighted average is
\beq\label{eq:weighted2}
\mathcal{T}_\omega^\ell (k_1) := \frac{1}{k_1} \sum_{j=1}^{k_1} E_{m_j}^{\ell} (\omega) \in I.
\eeq
We consider the tensor product operator $H_{\omega, \ell}^{k_1}$
on the Hilbert space $\mathcal{H}_\ell^{k_1} := \otimes_{k_1} \ell^2 (\Lambda_\ell)$ defined as
\beq\label{eq:tensor-hamiltonian1}
H_{\omega, \ell}^{k_1} := \frac{1}{k_1} \sum_{j=1}^{k_1} I \otimes \cdots \otimes H_{\omega, \ell} \otimes \cdots \otimes I,
\eeq
where the Hamiltonian $H_{\omega, \ell}$ appears in the $j^{th}$-position.
From the normalized eigenfunctions $\psi_m^{\ell}$ of $H_{\omega,\ell}$, we form the eigenfunctions
\beq\label{eq:tensore-ef1}
\psi_{\sigma(k_1)} := \psi_{m_1}^\ell \otimes \psi_{m_2}^\ell \otimes \cdots \otimes \psi_{m_{k_1}}^\ell ,
\eeq
where $\sigma(k_1) = (m_1, \ldots, m_{k_1})$.
These functions are eigenfunctions of $H_{\omega, \ell}^{k_1}$ with eigenvalue
$\mathcal{T}_\omega^\ell (k_1)$ defined in \eqref{eq:weighted2}:
\beq\label{eq:tensor-ev-eqn1}
 H_{\omega, \ell}^{k_1}  \psi_{\sigma(k_1)}^\ell  = \mathcal{T}_\omega^\ell (k_1)   \psi_{\sigma(k_1)}^\ell.
 \eeq

It is important to know how the eigenvalue average $\mathcal{T}_\omega^\ell (k_1)$ depends on the random variables $\omega_i, \omega_j$ with $i,j \in \Lambda_\ell$. As the matrix $H_{\omega, \ell}^{k_1}$ is of size $(m_k | \Lambda_\ell|)^{k_1} \times (m_k | \Lambda_\ell|)^{k_1}$, the expression ${\rm det} ( (m_k | \Lambda_\ell|)^{k_1} - E I)$ is a real polynomial of degree $m_k^{k_1}$ in the pair of random variables $\omega_i, \omega_j$ for $i,j \in \Lambda_\ell$.
We will use this in section \ref{sec:proof-prop1} when we apply the Harnack Curve Theorem.

\section{Proof of Proposition \ref{prop:decorrelation-prop1}}\label{sec:proof-prop1}
\setcounter{equation}{0}

In this section, we prove the technical result, Proposition \ref{prop:decorrelation-prop1}.
We let
$X_\ell(I_L(E)) := {\rm Tr} E_{H_{\omega,\ell}} (I_L(E))$,
$X_\ell(J_L(E^\prime)) := {\rm Tr} E_{H_{\omega,\ell}} (J_L(E^\prime))$, and consider the scale $\ell = [L^\alpha]$, for $0 < \alpha < 1$.
Then, we show
\beq\label{eq:same-number1}
\P \{ ( X_\ell(I_L(E)) \geq 1 ) \cap ( X_\ell(J_L(E')) \geq 1) \} \leq  \leq C_0 \frac{K \| \rho \|_\infty^2 m_k^{2 m_k} }{\Delta E - 4d}  \frac{ (C \log L)^{(1+\beta) d}}{L^{(2-\alpha) d}} .
\eeq
for positive numbers $0 < \alpha < 1$ and any $\beta > \frac{5}{2}$. 

%


\subsection{Reduction via the extended Minami estimate}\label{subsec:reduction1}

Let $\chi_A(\omega)$ be the characteristic function on the subset $A \subset \Omega$. In this section, we write $J_L(E) := L^{-d} J + E$ since we are dealing with one interval.
We use an extended Minami estimate of the form
\beq\label{eq:ext-minami1}
\E \{ \chi_{ \{ \omega ~|~  X_\ell(J_L(E)) \geq m_k +1 \} } X_\ell(J_L(E))( X_\ell(J_L(E)) - m_k ) \geq 1 \}  \nonumber \\
  \leq  C_M \left( \frac{\ell}{L} \right)^{2d} ,
\eeq
as follows from \cite{hislop-krishna1}.

\begin{lemma}\label{lemma:minami-consequence1}
Under the condition that the projectors have uniform dimension $m_k \geq 1$, we have
\beq\label{eq:ext-minami2}
\P \{ X_\ell(J_L(E)) > m_k \} \leq C_M \left( \frac{\ell}{L}\right)^{2d} .
\eeq
\end{lemma}

\begin{proof}
Recalling that $X_\ell(J_L(E)) \in \{ 0 \} \cup \N$, we have
\bea\label{eq:minami-conseq1}
\lefteqn{ \P \{ X_\ell(J_L(E)) > m_k \} } \nonumber \\
 & \leq & \P \{ X_\ell(J_L(E)) - m_k \geq 1 \} \nonumber \\
 & = &  \P \{ X_\ell(J_L(E)) ( X_\ell(J_L(E)) - m_k ) \geq 1 \}  \nonumber \\
 & = &  \P \{ \chi_{ \{ \omega ~|~  X_\ell(J_L(E)) \geq m_k +1 \} } X_\ell(J_L(E))( X_\ell(J_L(E)) - m_k ) \geq 1 \}  \nonumber \\
 & \leq & \E \{ \chi_{ \{ \omega ~|~  X_\ell(J_L(E)) \geq m_k +1 \} } X_\ell(J_L(E))( X_\ell(J_L(E)) - m_k ) \geq 1 \}  \nonumber \\
 & \leq & C_M \left( \frac{\ell}{L} \right)^{2d} ,
 \eea
by the extended Minami estimate.
\end{proof}

%

\subsection{Estimates on the joint probability}\label{subsec:joint-prob1}

We return to considering two scaled intervals $I_L(E)$ and $J_L(E^\prime)$, with $E \neq E^\prime$.
Because of \eqref{eq:ext-minami1}, we have
\bea\label{eq:same-number2}
\lefteqn{ \P \{ ( X_\ell(I_L(E)) \geq 1 ) \cap ( X_\ell(J_L(E')) \geq 1) \} } \nonumber \\
 & \leq & \P \{ ( X_\ell(I_L(E)) \geq m_k +1 ) \cap ( X_\ell(J_L(E^\prime)) \geq m_k +1 ) \} \nonumber \\
 & & +  \P \{ ( X_\ell(I_L(E)) \leq m_k ) \cap ( X_\ell(J_L(E^\prime)) \geq m_k +1 ) \}
 \nonumber \\
 & & +  \P \{ ( X_\ell(I_L(E)) \leq m_k + 1 ) \cap ( X_\ell(J_L(E^\prime)) \geq m_k ) \}
 \nonumber \\
 & & + \P \{ ( X_\ell(I_L(E)) \leq m_k  ) \cap ( X_\ell(J_L(E^\prime)) \leq m_k ) \}
 \nonumber \\
 & \leq & \P \{ ( X_\ell(I_L(E)) \leq m_k ) \cap ( X_\ell(J_L(E^\prime)) \leq m_k  ) \} \nonumber \\
 & &
 + C_0 \left( \frac{\ell}{L} \right)^{2d} .
\eea
The probability on the last line of \eqref{eq:same-number2} may be bounded above by
\bea\label{eq:reduction1}
\lefteqn{ \P \{ ( X_\ell(I_L(E)) \leq m_k ) \cap ( X_\ell(J_L(E^\prime)) \leq m_k  ) \} } & & \nonumber \\
 & \leq &
\sum_{k_1,k_2 = 1}^{m_k} \P \{ ( X_\ell(I_L(E))  = k_1 ) \cap ( X_\ell(J_L(E^\prime)) = k_2  ) \}.
\eea
Since $m_k$ is independent of $L$, it suffices to estimate
\beq\label{eq:same-number3}
\P \{ ( X_\ell(I_L(E))  = k_1 ) \cap ( X_\ell(J_L(E^\prime)) = k_2  ) \} ,
\eeq
for $1 \leq k_1, k_2 \leq m_k$.

\subsection{Reduction of length scale using localization}\label{subsec:localization1}

If we continue to work at scale $\ell = L^\alpha$, we can prove Proposition \eqref{prop:decorrelation-prop1}
but only for $\alpha > 0$ sufficiently small. In order to prove Proposition \ref{prop:decorrelation-prop1} for all $0 < \alpha < 1$, we must follow Klopp \cite[section 2.2]{klopp} and use localization in order to reduce the length scale to $\tilde{\ell} := C \log L$, for $C > 0$. Once we work with the length scale $\tilde{\ell}$, we will be able to prove Proposition \ref{prop:decorrelation-prop1} for all $0 < \alpha < 1$.
The goal of this section is to bound \eqref{eq:same-number3} by a similar estimate involving the length scale $\tilde{\ell}$ up to errors vanishing as $L \rightarrow \infty$. The key to this reduction is the localization properties of the Hamiltonians given in $(Loc)$.

\begin{defn}\label{defn:loc1}
We say that the local Hamiltonian $H_{\omega,\ell}$ satisfies $(Loc)$ in an interval $I \subset \Sigma$
if
\begin{enumerate}
\item The finite-volume fractional
moment criteria of \cite{aizenman1} holds on the interval $I$ for some constant $C>0$ sufficiently large;
\item There exists $\nu > 0$ such that, for any $p > 0$, there exists $q > 0$ and a length scale $\ell_0 > 0$ such
that, for all $\ell > \ell_0$, the following hold with probability greater than $1-L^{−p}$:
\begin{enumerate}
\item If $\varphi_{j}^\ell(\omega)$ is a normalized eigenvector of $H_{\omega, \ell}$ with eigenvalue $E_j^\ell (\omega) \in I$, and
\item $x_j^\ell(\omega)$ is a maximum of $x \rightarrow | \varphi_j^\ell (\omega)|$ in $\Lambda_\ell$,
\end{enumerate}
then, for $n \in \Lambda_\ell$, one has
\beq\label{eq:loc1}
|\varphi_j^\ell (\omega)(x)| \leq \ell^q e^{- \nu \| x - x_j^\ell(\omega)\|}.
\eeq
\end{enumerate}
The point $x_j^\ell(\omega)$ is called a localization center for $\varphi_j^\ell (\omega)$ or for $E_j^\ell (\omega)$.
\end{defn}

The main consequence in the present context of $(Loc)$ is the following result on the localization of eigenvectors.

\begin{lemma}\cite[Lemma 2.2]{klopp}\label{lemma:loc1}
Let $\Omega_0$ be the set of configuration $\omega$ for which $(Loc)$ holds with probability greater than $1 - L^{-2d}$.
There exists a covering of $\Lambda_\ell = \cup_{\gamma \in \Gamma} \Lambda_{\tilde{\ell}(\gamma)}$, where
$\Lambda_{\tilde{\ell}(\gamma)} = \Lambda_{\tilde{\ell}} + \gamma$, such that for all $\omega \in \Omega_0 \cap \{ \omega ~|~
(X_\ell (I_L(E)) = k_1 ) \cap (X_\ell (J_L(E^\prime)) = k_2 )\}$, we have that \\
\noindent
\textbf{either} \\
\noindent
(1) There exist $\gamma, \gamma^\prime \in \Gamma$ so that $\Lambda_{\tilde{\ell}(\gamma)} \cap \Lambda_{\tilde{\ell}(\gamma^\prime)} = \emptyset$ and
$$
X_{\tilde{\ell}(\gamma)} (\tilde{I}_L(E)) \geq 1 ~\mbox{and} ~
X_{\tilde{\ell}(\gamma^\prime)} (\tilde{J}_L(E^\prime)) \geq 1 ,
$$
\textbf{or}\\
\noindent
(2) For some $\gamma \in \Gamma$, we have $X_{5 \tilde{\ell}(\gamma)} (\tilde{I}_L(E)) = k_1$ and $X_{5 \tilde{\ell}(\gamma)} (\tilde{J}_L(E^\prime)) = k_2$.
\end{lemma}

Because of the localization properties of the eigenfunctions given in Lemma \ref{lemma:loc1}, we can reduce the estimate on scale $\ell$ to one on scale $\tilde{\ell}$ as presented in the following lemma.

\begin{lemma}\label{lemma:scale-reduction1}
For any $k_1, k_2 \in \{ 1, \ldots, m_k\}$, we have
\bea\label{eq:vanishing-prob1}
\lefteqn{\P \{ ( X_\ell(I_L(E))  = k_1 ) \cap ( X_\ell(J_L(E^\prime)) = k_2  ) \}} \nonumber \\
 &  \leq  &
\left( \frac{\ell}{L}  \right)^{2d} + \left(\frac{\ell}{\tilde{\ell}} \right)^{d} \P \{  (X_{5 \tilde{\ell}(\gamma)} (\tilde{I}_L(E)) = k_1 ) \cap (X_{5 \tilde{\ell}(\gamma)} (\tilde{J}_L(E^\prime)) = k_2) \}. \nonumber \\
 & &
\eea
\end{lemma}

%

\begin{proof}
According to Lemma \ref{lemma:loc1}, we have
\bea\label{eq:reduction-to-log1}
\lefteqn{ \P \{ ( X_{\ell}(I_L(E))  = k_1 ) \cap ( X_{\ell}(J_L(E^\prime)) = k_2  ) \} } \nonumber
\\ & \leq & \P \{ ( X_{\ell }(I_L(E))  = k_1 ) \cap ( X_{\ell }(J_L(E^\prime)) = k_2  ) \cap \Omega_0 \} + \P \{ \Omega \backslash \Omega_0 \} \nonumber \\
 & \leq & L^{-2d} + \P_1 + \P_2 ,
 \eea
where, according to Lemma \ref{lemma:loc1}, $\P_1$ is the probability that option (1) occurs and $\P_2$ is the probability that option (2) occurs.
To estimate $\P_1$, we use the independence of the Hamiltonians associated with $\Lambda_{\tilde{\ell}(\gamma)}$ and  $\Lambda_{\tilde{\ell}(\gamma^\prime)}$, together with Wegner's estimate on scale $\tilde{\ell}$, to obtain
\bea\label{eq:reduction-to-log2}
\P_1 & \leq & \left(\frac{\ell}{\tilde{\ell}}   \right)^{2d} \P \{ ( X_{\tilde{\ell}(\gamma)}(\tilde{I}_L(E)) \geq 1 ) \cap ( X_{\tilde{\ell}(\gamma^\prime)}(\tilde{J}_L(E^\prime)) \geq 1) \} \nonumber \\
 & \leq & \left(\frac{\ell}{\tilde{\ell}}   \right)^{2d} \P \{  X_{\tilde{\ell}(\gamma)}(\tilde{I}_L(E))   \geq 1 \} ~ \P \{ X_{\tilde{\ell}(\gamma^\prime)}(\tilde{J}_L(E^\prime)) \geq 1 \} \nonumber \\
 & \leq & C_W^2 \left(\frac{\ell}{\tilde{\ell}}   \right)^{2d} ( | \tilde{I}_L(E) | \tilde{\ell}^d )(| \tilde{I}_L(E) | \tilde{\ell}^d) \nonumber \\
 & \leq & C_W^2 \left(\frac{\ell}{L}   \right)^{2d} .
\eea
In order to estimate $\P_2$, condition (2) implies
\beq\label{eq:reduction-to-log3}
\P_2  \leq \left(\frac{\ell}{\tilde{\ell}} \right)^{d} \P \{  (X_{5 \tilde{\ell}(\gamma)} (\tilde{I}_L(E)) = k_1 ) \cap (X_{5 \tilde{\ell}(\gamma)} (\tilde{J}_L(E^\prime)) = k_2) \} .
\eeq
Bounds  \eqref{eq:reduction-to-log1}--\eqref{eq:reduction-to-log3} imply the bound \eqref{eq:vanishing-prob1}.
\end{proof}



\subsection{Key estimate on the $\log$-scale}\label{subsec:smallest-length-scale1}

\noindent
The proof of the next key Proposition \ref{prop:vanishing-prob1} follows the ideas in \cite{klopp}.

\begin{pr}\label{prop:vanishing-prob1}
For $k_1, k_2 = 1, \ldots, m_k$ and for any $\beta > \frac{5}{2}$, there exists a scale $L_0 > 0$, so that for any
 $L > L_0$, there exists a constant $C_0> 0$ so that
, we have
\beq\label{eq:vanishing-prob2}
\P \{ ( X_{5\tilde{\ell}(\gamma)}(\tilde{I}_L(E))  = k_1 ) \cap ( X_{5\tilde{\ell}(\gamma)}(\tilde{J}_L(E^\prime)) = k_2  ) \} \leq
C_0 \frac{K \| \rho \|_\infty^2 m_k^{2 m_k} }{\Delta E - 4d}  \frac{ (C \log L)^{(2+\beta) d}}{L^{2 d}} .
\eeq
\end{pr}

\begin{proof}
1. We begin with some observation concerning the eigenvalue averages.
We let $\Omega_0(\tilde{\ell}, k_1, k_2)$ denote the event
\beq\label{eq:joint-prob1}
\Omega_0(\tilde{\ell}, k_1, k_2) := \{ \omega ~|~  (X_{\tilde{\ell}}((\tilde{I}_L(E))) = k_1 ) \cap ( X_{\tilde{\ell}}((\tilde{J}_L(E^\prime))) = k_2) \} \cap \Omega_0 ,
\eeq
for $k_1, k_2 = 1, \ldots, m_k$, and where we write $\tilde{\ell}$ instead of $5 \tilde{\ell}(\gamma)$ to simplify the notation.
We define the subset $\Delta \subset \Lambda_{\tilde{\ell}} \times \Lambda_{\tilde{\ell}}$ by $\Delta := \{ (i,i) ~|~ i \in \Lambda_{\tilde{\ell}} \}$. For each pair of sites $(i,j) \in \Lambda_{\tilde{\ell}} \times \Lambda_{\tilde{\ell}} \backslash \Delta$, the Jacobian determinant of the mapping $\varphi: (\omega_i, \omega_j) \rightarrow
( \mathcal{T}_{\tilde{\ell}} (E, k_1), \mathcal{T}_{\tilde{\ell}} (E^\prime, k_2))$, given by:
\beq\label{eq:jacobian1}
     J_{ij}( \mathcal{T}_{\tilde{\ell}} (E,k_1), \mathcal{T}_{\tilde{\ell}}(E^\prime, k_2) ) :=     \left|    \begin{array}{cc}
                   \partial_{\omega_i}\mathcal{T}_{\tilde{\ell}} (E, k_1) & \partial_{\omega_j} \mathcal{T}_{\tilde{\ell}} (E, k_1) \\
                   \partial_{\omega_i} \mathcal{T}_{\tilde{\ell}} (E^\prime, k_2) & \partial_{\omega_j} \mathcal{T}_{\tilde{\ell}} (E^\prime, k_2)
                    \end{array}   \right| .
\eeq
As we will show in section \ref{subsec:diffeom1}, the condition $J_{ij}( \mathcal{T}_{\tilde{\ell}} (E,k_1), \mathcal{T}_{\tilde{\ell}}(E^\prime, k_2) ) \geq \lambda (L) > 0$ implies that the average of the eigenvalues in $\tilde{I}_L(E)$ and $\tilde{J}_L(E^\prime)$ effectively vary independently with respect to any pair of independent random variables $(\omega_i, \omega_j)$, for $i \neq j$. We define the following events for pairs $(i,j) \in \Lambda_{\tilde{\ell}} \times \Lambda_{\tilde{\ell}} \backslash \Delta$:
\beq\label{eq:pair-event1}
\Omega_0^{i,j} ({\tilde{\ell}}, k_1, k_2) := \Omega_0({\tilde{\ell}}, k_1, k_2) \cap \{ \omega ~|~ J_{ij}( \mathcal{T}_{\tilde{\ell}} (E,k_1), \mathcal{T}_{\tilde{\ell}}(E^\prime, k_2) ) \geq \lambda (L) \},
\eeq
where $\lambda (L) > 0$ is given by
\beq\label{eq:defn-lambda1}
\lambda(L) :=     ( \Delta E - 4d) K^{-1} (C \log L)^{- \beta d}, 
\eeq
where the exponent $\beta > 0$ is chosen below.
Following Klopp \cite[pg.\ 242]{klopp}, we note in section \ref{subsec:diffeom1} that the positivity of the Jacobian determinant insures that the map $\varphi$, restricted to a certain domain, is a diffeomorphism. In particular, this allows us to compute $\P \{ \Omega_0^{i,j} ({\tilde{\ell}}, k_1, k_2) \}$ as in Lemma \ref{lemma:diffeom-new1}.

\noindent
2. We next bound $\P \{ \Omega_0 ({\tilde{\ell}}, k_1, k_2) \}$ in terms of $\P \{ \Omega_0^{i,j} ({\tilde{\ell}}, k_1, k_2) \}$ using
\cite[Lemma 2.5]{klopp}. This lemma states that for $(u,v ) \in (\R^+)^{2n}$ normalized so that $\| u \|_1 = \| v \|_1 = 1$, we have
\beq\label{eq:jac-lb1}
\max_{j \neq k} \left| \begin{array}{cc}
                            u_j & u_k \\
                            v_j & v_k
                            \end{array} \right|^2 \geq \frac{1}{4 n^5} \| u-v \|_1^2.
\eeq
Applying this with $n = (2 \ell + 1)^d$, and  $u = \nabla_\omega \mathcal{T}(\omega)$ and $v = \nabla_{\omega} \mathcal{T}^\prime (\omega)$,
and recalling the positivity property mentioned in point (2) following from \eqref{eq:tr-deriv2}
and the normalization in point (3) of section \ref{subsec:weight-trace1}, we obtain from \eqref{eq:jac-lb1} and \eqref{eq:variation5}:
\bea\label{eq:jac-lb2}
\max_{i \neq j \in \Lambda_{\tilde{\ell}}} J_{ij} ( \mathcal{T}_{\tilde{\ell}} (E), \mathcal{T}_{\tilde{\ell}} (E^\prime))^2  &\geq & \left( \frac{2^3}{ {\tilde{\ell}}^{5d}} \right)
  \| \nabla_\omega ( \mathcal{T}_{\tilde{\ell}} (E) - \mathcal{T}_{\tilde{\ell}} (E^\prime)) \|_1^2 \nonumber \\
 & \geq &
\left( \frac{\Delta E - 4d}{K} \right)^2 \left( \frac{2^3}{ {\tilde{\ell}}^{5d}} \right) .
\eea
We partition the probability space as $\{ \omega ~|~ J_{ij} \geq \lambda(L) ~{\rm some} ~(i,j) \in \Lambda_{\tilde{\ell}} \times \Lambda_{\tilde{\ell}} \backslash \Delta \} \cup \{ \omega ~|~ J_{ij} < \lambda(L) ~\forall ~(i,j) \in \Lambda_{\tilde{\ell}} \times \Lambda_{\tilde{\ell}} \backslash \Delta \}$, where we write $J_{ij}$ for the Jacobian $J_{ij} ( \mathcal{T}_{\tilde{\ell}} (E), \mathcal{T}_{\tilde{\ell}} (E^\prime))$.
Suppose that the second event $\{ \omega ~|~ J_{ij} < \lambda(L) ~\forall ~(i,j) \in \Lambda_{\tilde{\ell}} \times \Lambda_{\tilde{\ell}} \backslash \Delta \}$ occurs, so that from \eqref{eq:jac-lb2}, we have:
\bea\label{eq:jac-lb3}
\lefteqn{ \lambda(L)^2 = \left( \frac{\Delta E - 4 d}{K (C \log L)^{\beta d}} \right)^2 \geq  \max_{i \neq j \in \Lambda_{\tilde{\ell}}} J_{ij} ( \mathcal{T}_{\tilde{\ell}} (E), \mathcal{T}_{\tilde{\ell}} (E^\prime))^2 } \nonumber \\
  &\geq & \left( \frac{2^3}{ {\tilde{\ell}}^{5d}} \right)
  \| \nabla_\omega ( \mathcal{T}_{\tilde{\ell}} (E) - \mathcal{T}_{\tilde{\ell}} (E^\prime)) \|_1^2 .
\eea
Taking ${{\tilde{\ell}}} = C \log L$, this implies that
\beq\label{eq:gradients1}
    \| \nabla_\omega ( \mathcal{T}_{\tilde{\ell}} (E) - \mathcal{T}_{\tilde{\ell}} (E^\prime)) \|_1^2 \leq C_1 (C \log L)^{(5 - 2 \beta)d}. 
\eeq
So, provided $\beta > 5/2$, 
we find that the bound \eqref{eq:gradients1} implies that the $\nabla_\omega \mathcal{T}_{\tilde{\ell}} (E)$ is almost collinear with $\nabla_\omega \mathcal{T}_{\tilde{\ell}} (E^\prime)$. This contradicts the lower bound \eqref{eq:variation5} as long as $\Delta E - 4d > 0$. Consequently, the probability of the second event is zero.

\noindent
3. It follows from this, the partition of the probability space, and Lemma \ref{lemma:diffeom-new1} that
\bea\label{eq:independent-est2}
 \P \{ \Omega_0 (\tilde{\ell}, k_1, k_2) \} & \leq & \sum_{(i,j) \in \Lambda_{\tilde{\ell}} \times \Lambda_{\tilde{\ell}} \backslash \Delta} \P \{ \Omega_0^{i,j} (\tilde{\ell}, k_1, k_2) \} \nonumber \\
  & \leq & {\tilde{\ell}}^{2d} C_0 \frac{K \| \rho \|_\infty^2 m_k^{2 m_k} }{\Delta E - 4d}  \frac{ (C \log L)^{\beta d}}{L^{2 d}} .
  \eea
With $\tilde{\ell} = C \log L$,
we find the probability
$\P \{ \Omega_0 (\tilde{\ell}, k_1, k_2) \}$ is bounded as
 \beq\label{eq:independent-est3}
  \P \{ \Omega_0 (\tilde{\ell}, k_1, k_2) \} \leq
  C_0 \frac{K \| \rho \|_\infty^2 m_k^{2 m_k} }{\Delta E - 4d}  \frac{ (C \log L)^{(2+\beta) d}}{L^{2 d}} ,
  \eeq
Replacing $\tilde{\ell}$ by $5 {\tilde{\ell}}(\gamma)$, changing the constant $C_0$,  this completes the proof of Proposition \ref{prop:vanishing-prob1}.
\end{proof}

In summary, Proposition \ref{prop:vanishing-prob1} shows that
$$
\P \{ \Omega_0 (5\tilde{\ell}(\gamma), k_1, k_2) \} \rightarrow 0, {\rm as} ~~L \rightarrow 0.
$$
As a consequence of this and \eqref{eq:reduction-to-log1}--\eqref{eq:reduction-to-log3}, there exist constants $C_0, C_1 > 0$ such that for all $L >>0$,
\bea\label{eq:vanishing4}
\lefteqn{\P \{ (X_\ell(I_L(E)) = k_1 ) \cap ( X_\ell(J_L(E^\prime)) = k_2) \} } \nonumber \\
 &\leq & C_1 \left( \frac{1}{L^{1-\alpha}} \right){2d} + C_0 \frac{K \| \rho \|_\infty^2 m_k^{2 m_k} }{\Delta E - 4d}  \frac{ (C \log L)^{(1+\beta) d}}{L^{(2-\alpha) d}},
\eea
showing that $\P \{ (X_\ell(I_L(E)) = k_1 ) \cap ( X_\ell(J_L(E^\prime)) = k_2) \}  \rightarrow 0$ as $L \rightarrow \infty$,
for any integers $k_1, k_2 = 1, \ldots, m_k$.
This proves, up to the proof of the diffeomorphism property of $\varphi$, the main result \eqref{eq:prob1}.

\subsection{Estimate of $\P \{ \Omega_0^{i,j} (\tilde{\ell}, k_1, k_2) \}$}
\label{subsec:diffeom1}

Let $\Omega_0(\tilde{\ell}, k_1, k_2), k_1, k_2 = 1, \ldots, m_k$ be the set of configurations described in \eqref{eq:joint-prob1}. Similarly, for any pair of sites $(i,j) \in \Lambda_{\tilde{\ell}}
\times \Lambda_{\tilde{\ell}} \backslash \Delta$, the Jacobian determinant
$J_{ij}( \mathcal{T}_{\tilde{\ell}} (E,k_1), \mathcal{T}_{\tilde{\ell}}(E^\prime, k_2) )$ is defined in equation (\ref{eq:jacobian1}).
We also defined events $\Omega_0^{i,j} ({\tilde{\ell}}, k_1, k_2)$, for pairs $(i,j) \in \Lambda_{\tilde{\ell}} \times \Lambda_{\tilde{\ell}} \backslash \Delta$, in \eqref{eq:pair-event1}:
\beq\label{eq:pair-event2}
\Omega_0^{i,j} ({\tilde{\ell}}, k_1, k_2) := \Omega_0({\tilde{\ell}}, k_1, k_2) \cap \{ \omega ~|~ J_{ij}( \mathcal{T}_{\tilde{\ell}} (E,k_1), \mathcal{T}_{\tilde{\ell}}(E^\prime, k_2) ) \geq \lambda(L) \},
\eeq
where $\lambda (L) > 0$ has the value
\beq\label{eq:lambda1}
\lambda (L) := \frac{\Delta E - 4d}{K ( C \log L)^{ d \beta}} ,
\eeq
for some $\beta > \frac{5}{2}$.
We present an important technical lemma that is a simplification of \cite[Lemma 2.6]{klopp}.

\begin{lemma}\label{lemma:diffeom-new1} 
For all $L$ large, there is a finite constant $C_0 > 0$, independent of $L$, so that
$$
\P \{ \Omega_0^{ij} ( {\tilde{\ell}}, k_1, k_2) \} \leq  C_0 \frac{K \| \rho \|_\infty^2 m_k^{2 m_k} }{\Delta E - 4d}  \frac{ (C \log L)^{\beta d}}{L^{2 d}},
$$
for any $\beta > \frac{5}{2}$.
\end{lemma}

\begin{proof}
\noindent
1. Since $\Omega_0^{ij} ( {\tilde{\ell}}, k_1, k_2) \subset \Omega$ is measurable and bounded, we can find a compact set $K \subset \Omega_0^{ij} ({\tilde{\ell}}, k_1, k_2 )$ so that
\beq\label{eq:int-reg1}
\P \{ \Omega_0^{i,j}({\tilde{\ell}}, k_1, k_2)  \backslash K \} \leq C_1 L^{- 2d} .
\eeq
We define the map $\varphi : \Omega_0^{i,j}({\tilde{\ell}}, k_1, k_2) \rightarrow \R^2$ by
$$
\varphi (\omega_i, \omega_j) = ( \mathcal{T}_{\tilde{\ell}} (E, k_1, \omega), \mathcal{T}_{\tilde{\ell}} (E^\prime, k_2,\omega)) .
$$
This map is continuous so $\varphi (K) \subset \tilde{I}_L(E) \times \tilde{J}_L(E^\prime)$ is compact.

\noindent
2. For each $p \in \varphi (K)$, we choose any element $\omega_{ij}(p) \in K$ in the pre-image of $p$ under $\varphi^{-1}$: $\omega_{ij}(p) \in \varphi^{-1} (p) \in K \subset \Omega_0^{i,j}({\tilde{\ell}}, k_1, k_2)$. Because the Jacobian of $\varphi$ is bounded below at each point of $\varphi^{-1}(p)$, as follows from the definition of $\Omega_0^{i,j}({\tilde{\ell}}, k_1, k_2)$, the Inverse Function Theorem states that there are open balls $U_{\omega_{ij}(p)} \subset \Omega_0^{i,j}({\tilde{\ell}}, k_1, k_2)$ and $V_p \subset \tilde{I}_L(E)\times \tilde{J}_L(E^\prime)$, with $\omega_{ij}(p) \in U_{\omega_{ij}(p)}$ and $p \in V_p$  so that the restriction $\varphi$ to $U_{\omega_{ij}(p)}$ is a diffeomorphism with $V_p$.
As a consequence, the point $\omega_{ij}(p) \in K$ is the unique point in $U_{\omega_{ij}(p)}$ such that $\varphi (\omega_{ij}) = p$ so that such points are isolated points of $K$. It follows that $\varphi^{-1}(p)$ is a discrete subset of $K$.

\noindent
3. We can apply Harnack's Curve Theorem \cite{Harnack1}
in order to obtain an upper bound on the number of points in $\varphi^{-1}(p)$. In section \ref{subsec:dependence-rv1}, we showed that $\mathcal{T}_\omega^{\tilde{\ell}} ( k_1)$ is a zero of a determinant constructed from the tensor product operator $H_{\omega,{\tilde{\ell}}}^{k_1}$. For fixed $E \in \R$, this determinant defines the function $f_E^{k_1}(\omega_i, \omega_j) := \det ( H_{\omega,{\tilde{\ell}}}^{k_1} - E)$ that is a polynomial of degree at most $m_k^{k_1}$ in each random variable $\omega_i$ and $\omega_j$. As such, the polynomial  $f_E^{k_1}(\omega_i, \omega_j)$ may be extended to $\R^2$. For two distinct energies $e \neq e^\prime$, we consider the polynomial $f_{e,e^\prime}^{k_1,k_2}(\omega_1, \omega_2) := [f_e^{k_1}(\omega_i, \omega_j)]^2 + [f_{e^\prime}^{k_1}(\omega_i, \omega_j) ]^2$.  The Harnack Curve Theorem states that the maximum number of connected components of the zero set of $f_{e,e^\prime}^{k_1,k_2}$ is bounded above by $\max (m_k^{2 k_1}, m_k^{2 k_2}) $.
Each of those connected components lying in $K$ is necessarily zero dimensional by the above diffeomorphism argument.
Hence, since $k_1, k_2 \leq m_k$, the number of points in the set $\varphi^{-1}(p)$ in $K$ is bounded above by $m_k^{2 m_k}$, independent of $L$. 

\noindent
4. The sets $\{ V_p \}_{p \in \varphi (K)}$ cover $\varphi (K)$. Since $\varphi (K)$ is compact, there is a finite subcover $\{ V_{p_{t}} \}_{t = 1}^N$ so that
$$
\varphi (K) \subset \cup_{t = 1}^N V_{p_t} \subset \tilde{I}_L(E) \times \tilde{J}_L(E^\prime ).
$$
The restriction of $\varphi$ to $U_{\omega_{ij}(p_t)}$, a diffeomorphism with $V_{p_t}$, is denoted by $\varphi_{p_t}$.
We take intersections and relative complements to obtain a finite collection $\{ W_m \}$ of {\it disjoint} sets so that
$$
\varphi (K) \subset \cup_{m = 1}^{\tilde{N}} W_m = \cup_{t = 1}^{N} V_{p_t},
$$
up to a set of Lebesgue measure zero, and where $\tilde{N}$ is a function of $N$, and each $W_m \subset V_{p_{t}}$, for some index $t$. We can choose
$W_m$ so that it is in the domain of $\varphi_{{p_t}}^{-1}$.

\noindent
5. We compute the $\P$-measure of $\Omega_0^{i,j}({\tilde{\ell}}, k_1, k_2)$ by first computing the measure of $\varphi_{p_t}^{-1} ( W_m)$:
\beq\label{eq:volume1}
\P \{ \varphi_{p_t}^{-1} ( W_m) \} = \int_{\varphi_{p_t}^{-1}(W_m)} ~\rho (\omega_i) \rho(\omega_j) ~d \omega_i ~d \omega_j.
\eeq
Upon changing variables, we obtain
\beq\label{eq:volume2}
\int_{\varphi_{p_t}^{-1}(W_m)} ~\rho (\omega_i) \rho (\omega_j) ~d  \omega_i ~d \omega_j
\leq | {\rm Jac} \varphi_{p_t}^{-1}| \| \rho\|_\infty^2 \int_{W_m} ~d E ~d {E^\prime} .
\eeq
It follows from \eqref{eq:lambda1} that the Jacobian $| {\rm Jac} \varphi_{p_t}^{-1}|$ satisfies the bound
\beq\label{eq:inv-jac11}
| {\rm Jac} \varphi_{p_t}^{-1}(\mathcal{T}_{\tilde{\ell}} (\omega_{ij}), \mathcal{T}_{\tilde{\ell}}^\prime (\omega_{ij}))| \leq \frac{K (C \log L)^{\beta d}}{\Delta E - 4d}.
\eeq
Using this bound, the fact that the $W_m$ are disjoint, and the fact that there are at most $m_k^{2 m_k}$ isolated points in $\varphi^{-1}(p)$, for $p \in \varphi (K)$, we find
\bea\label{eq:volume3}
\P \{ K \} & \leq & \| \rho \|_\infty^2 m_k^{2 m_k} \left( \max_{n} | {\rm Jac} \varphi_n^{-1} | \right) \left[ \sum_{m = 1}^{\tilde{N}} |  W_m |   \right]  \nonumber \\
 & \leq & \frac{K (C \log L)^{\beta d}}{\Delta E - 4d} ~ \rho \|_\infty^2 | ~m_k^{2 m_k} ~\| \tilde{I}_L(E) \times \tilde{J}_L(E^\prime) | \nonumber \\
  & \leq &  \frac{K \| \rho \|_\infty^2 m_k^{2 m_k} }{\Delta E - 4d}  \frac{ (C \log L)^{\beta d}}{L^{2 d}}
\eea
Finally, we have
\beq\label{eq:volume4}
\P \{ \Omega_0^{i,j} ({\tilde{\ell}}, k_1, k_2) \} \leq \P \{ K \} + \P \{ \Omega_0^{i,j} ({\tilde{\ell}}, k_1, k_2) \backslash K  \},
\eeq
so the result follows from \eqref{eq:volume3} and the fact that $\P \{ \Omega_0^{i,j} ({\tilde{\ell}}, k_1, k_2) \backslash K  \}$ is $\mathcal{O}(L^{-2d})$.
\end{proof}

\section{Asymptotically independent random variables: Proof of Theorem \ref{thm:decorrelation-lattice1}}\label{sec:proof1}
\setcounter{equation}{0}

In this section, we give the proof of Theorem \ref{thm:decorrelation-lattice1}.
To prove that $\xi_E^\omega(I)$ and $\xi_{E^\prime}^\omega (J)$ are independent, we recall that the limit points $\xi_E^\omega$ are the same as
those obtained from a certain uniformly asymptotically negligible array (\cite[Proposition 4.4]{hislop-krishna1}). To obtain this array,
we construct a cover of $\Lambda_L$ by non-overlapping cubes of side length $2 \ell + 1$ centered at points $n_p$. We use $\ell = [L^\alpha]$,
where $(\alpha, \beta)$ satisfy $0 < \alpha < 1$ and $\beta > \frac{5}{2}$.
The number of such cubes $\Lambda_\ell (n_p)$ is $N_L := [ (2 L+1) / (2\ell +1)]^d$. The local Hamiltonian is $H^\omega_{p, \ell}$. The associated eigenvalue point process at energy $E$ is denoted by $\eta_{E,\ell,p}^\omega$. We define the point process $\zeta^\omega_{E,\Lambda_L} = \sum_{p=1}^{N_L} \eta^\omega_{E,p, \ell}$.
For a bounded interval $I \subset \R$, we define the local random variable $\eta^\omega_{E,\ell,p}(I) := {\rm Tr} ( E_{H^\omega_{p,\ell}}(I_L(E)))$ and similarly $\eta^\omega_{E^\prime,\ell,p}(J)$ for the scaled interval $J_L(E^\prime)$.
For $p \neq p'$, the random variables $\eta^\omega_{E,\ell,p}(I)$ and $\eta^\omega_{E^\prime,\ell,p^\prime}(J)$ are independent for any energies $E$ and $E^\prime$ and any bounded intervals $I$ and $J$.
We compute
\bea\label{eq:independent1}
\P \{ (\zeta^\omega_{E,\Lambda_L}(I) \geq 1) \cap   (\zeta^\omega_{E^\prime \Lambda_L}(J) \geq 1) \} &=& \sum_{p,p^\prime = 1}^{N_L} \P \{ (\eta^\omega_{E, \ell,p}(I) \geq 1) \cap   (\eta^\omega_{E^\prime, \ell,p^\prime}(J) \geq 1) \} \nonumber \\
 &= & \sum_{p, p^\prime = 1}^{N_L} \P \{ \eta^\omega_{E,\ell,p}(I) \geq 1 \} \P \{  \eta^\omega_{E^\prime, \ell,p^\prime}(J) \geq 1 \} \nonumber \\
 & & + \mathcal{E}_L(E, E^\prime, I, J) ,
 \eea
where the error term is just the diagonal $p=p'$ contribution:
\bea\label{eq:error1}
\mathcal{E}_L(E, E^\prime, I, J)  & = &   \sum_{p = 1}^{N_L} \left[ \P \{ (\eta^\omega_{E,\ell,p}(I) \geq 1) \cap
 (\eta^\omega_{E^\prime, \ell,p}(J) \geq 1) \} \right.             \nonumber \\
 &  & \left. -
 \P \{ \eta^\omega_{E,\ell,p}(I) \geq 1 \} \P \{ \eta^\omega_{E^\prime,\ell,p}(J) \geq 1 \} \right].
 \eea

If we now assume that $|E - E^\prime| > 4d$ and $E, E^\prime \in \Sigma_{\rm CL}$, then the first term on the
right side of \eqref{eq:error1} is bounded above by $L^{-d} (\log L)^{(1+\beta)d}$ due to the decorrelation estimate \eqref{eq:prob1}.
The bound on the second probability on the right of \eqref{eq:error1} is $C_W^2 L^{-2d (1-\alpha)}$. It is obtained from the square of the Wegner estimate
$$
\P \{ \eta^\omega_{E^\prime, \ell, p}(J) \geq 1 \} \leq C_W (\ell / L)^d = C_W L^{-d(1 - \alpha)}.
$$is bounded
Since $N_L \sim (L / \ell)^d = L^{(1- \alpha)d}$, we find that the second term
on the right of \eqref{eq:error1} above by $C_W^2 L^{-d(1- \alpha)}$.
Consequently, the error term $\mathcal{E}_L(E, E^\prime, I, J) \rightarrow 0$ as $L \rightarrow \infty$.
Since the set of limit points $\zeta^\omega$ and $\xi^\omega$ are the same \cite{hislop-krishna1}, this estimate proves that
\beq\label{eq:independent2}
\lim_{L \rightarrow \infty}   \P \{ (\zeta^\omega_{E,\Lambda_L}(I) \geq 1) \cap   (\zeta^\omega_{E^\prime,\Lambda_L}(J) \geq 1) \} = \P \{ \xi^\omega_E(I) \geq 1 \} \P \{ \xi^\omega_{E^\prime} (J) \geq 1 \} ,
\eeq
establishing the asymptotic independence of the random variables $\xi^\omega_E(I)$ and $\xi^\omega_{E^\prime}(J)$ provided $|E - E^\prime | > 4d$.

\section{Bounds on eigenvalue multiplicity}\label{sec:multiplicity1}
\setcounter{equation}{0}

The extended Minami estimate may be used with the Klein-Molchanov argument \cite{klein-molchanov} to bound the multiplicity of eigenvalues in the localization regime. The basic argument of Klein-Molchanov is the following. If $H_\omega$ has at least $m_k +1$ linearly independent eigenfunctions with eigenvalue $E$ in the localization regime, so that the eigenfunctions exhibit rapid decay, then any finite volume operator $H_{\omega, L}$ must have at least $m_k + 1$ eigenvalues close to $E$ for large $L$. But, by the extended Minami estimate, this event occurs with small probability.
The first lemma is a deterministic result based on perturbation theory.

\begin{lemma}\label{lemma:local-ef-est1}
Suppose that $E \in \sigma (H)$ is an eigenvalue of a self adjoint operator $H$ with multiplicity at least $m_k + 1$. Suppose that all the associated eigenfunctions decay faster than $\langle x \rangle^{- \sigma}$, for some $\sigma > d / 2 > 0$. We define $\epsilon_L := C L^{- \sigma + \frac{d}{2}}$. Then for all $L>>0$, the local Hamiltonian $H_{L} := \chi_{\Lambda_L} H \chi_{\Lambda_L}$ has at least $m_k + 1$ eigenvalues in the interval $[E-  \epsilon_L, E+  \epsilon_L]$.
\end{lemma}

\begin{proof}
1. Let $\{ \varphi_j ~|~ j = 1, \ldots , M \}$ be an orthonormal basis of the eigenspace for $H$ and eigenvalue $E$. We assume that the eigenvalue multiplicity $M \geq m_k + 1$. We define the local functions $\varphi_{j , L} := \chi_{\Lambda_L} \varphi_j$, for $j = 1 , \ldots, M$.
These local functions satisfy:
\bea\label{eq:local-ef1}
1 - \epsilon_L & \leq & \| \varphi_{j,L}  \|  \leq   1 , \nonumber \\
 | \langle \varphi_{i,L} , \varphi_{j, L} \rangle | & \leq & \epsilon_L, ~~~~i \neq j .
\eea
It is easy to check that these conditions imply that the family is linearly independent. Let $V_L$ denote the $M$-dimensional subspace of $\ell^2 (\Lambda_L)$ spanned by these functions.

\noindent
2. As in \cite{klein-molchanov}, it is not difficult to
prove that the functions $\varphi_{j,L}$ are approximate eigenfunctions for $H_L$:
\beq\label{eq:approx-ef1}
\| (H_L - E) \varphi_{j,L} \| \leq \epsilon_L \| \varphi_{j,L} \| .
\eeq
Furthermore, for any $\psi_L \in V_L$, we have $\| (H_L - E) \psi_L \| \leq 2 \epsilon_L \| \psi_L \|$.

\noindent
3. Let $J_L := [ E - 3 \epsilon_L , E+ 3 \epsilon_L ]$. We write $P_L$ for the spectral projector $P_L :=
\chi_{J_L}(H_L)$ and $Q_L := 1 - P_L$ is the complementary projector. For any $\psi \in V_L$, we have
$\| Q_L \psi \| \leq (3 \epsilon_L)^{-1} \| (H_L - E) Q_L \psi \| \leq (2/3) \| \psi \|.$
Since $\|P_L \psi \|^2 = \| \psi \|^2 - \| Q_L \psi \|^2 \geq (5/9) \| \psi \|$,
it follows that $P_L: V_L \rightarrow \ell^2 (\Lambda_L)$ is injective. Consequently, we have
$$
{\rm dim} \Ran P_L = {\rm Tr} (P_L ) \geq {\rm dim} ~V_L = M > m_k.
$$
Redefining the constant $C >0$ in the definition of $\epsilon_L$, we find that
$H$ has at least $m_k + 1$ eigenvalues in $[E- \epsilon_L, E+ \epsilon_L]$.
\end{proof}

The second lemma is a probabilistic one and the proof uses the extended Minami estimate.

\begin{lemma}\label{lemma:local-ef-est2}
Let $I \subset \R$ be a bounded interval. For $q > 2d$,  and any interval
$J \subset I$ with $|J| \leq L^{-q}$, we define the event
\beq\label{eq:event1}
\mathcal{E}_{L,I,q} := \{ \omega ~|~ {\rm Tr} ( \chi_{J} (H_{\omega, L})) \leq m_k ~\forall J \subset I, |J| \leq L^{-q} \} .
\eeq
Then, the probability of this event satisfies
\beq\label{eq:prob1-1}
\P \{ \mathcal{E}_{L,I,q} \} \geq 1 - C_0 L^{2d-q} .
\eeq
\end{lemma}

\begin{proof}
 We cover the interval $I$ by $2 ( [ L^q |I| /2] +1)$ subintervals of length $2 L^{-q}$ so that any subinterval $J$ of length $L^{-q}$ is contained in one of these. We then have
\beq\label{eq:minami-prob1}
\P \{ \mathcal{E}_{L,I,q}^{\rm c} \} \leq (L^q |I| + 2) \P  \{ \chi_J(H_{\omega,L}) > m_k \}.
\eeq
The probability on the right side is estimated from the extended Minami estimate
\beq\label{eq:minami-prob2}
  \P  \{ \chi_J(H_{\omega,L}) > m_k \} \leq C_M (L^{-q} L^d)^2 = C_M L^{2(d-q)},
  \eeq
so that
\beq\label{eq:minami-prob3}
\P \{ \mathcal{E}_{L,I,q}^{\rm c} \} \leq C_M (L^q |I| + 2) L^{2(d-q)} = C_M (|I| +1)L^{2d-q}.
\eeq
This establishes \eqref{eq:prob1-1}.
\end{proof}

\begin{thm}\label{thm:multiplicity1}
Let $H^\omega$ be the generalized Anderson Hamiltonian described in section \ref{sec:introduction} with perturbations $P_i$ having uniform rank $m_k$.  Then the eigenvalues in the localization regime have multiplicity at most $m_k$ with probability one.
\end{thm}

\begin{proof}
We consider a length scale $L_k = 2^k$. It follows from \eqref{eq:prob1-1} that
the probability of the complementary event $\mathcal{E}_{L_k,I,q}^{\rm c}$
is summable. By the Borel-Cantelli Theorem, that means  for almost every $\omega$ there is a $k(q, \omega)$ so that for all $k > k(q, \omega)$
the event $\mathcal{E}_{L_k,I,q}$ occurs with probability one. Let us suppose that $H^\omega$ an eigenvalue with multiplicity at least $m_k +1$ in an interval $I$ and that the corresponding eigenfunctions decay exponentially. Then, by Lemma \ref{lemma:local-ef-est1}, the local Hamiltonian $H_{\omega, L_k}$ has at least $m_k + 1$ eigenvalues in the interval $[E- \epsilon_L, E+ \epsilon_L]$ where $\epsilon_L = CL^{-(\beta - \frac{d}{2})}$, for any $\beta > 5 d/2$.  This contradicts the event $\mathcal{E}_{L_k,I,q}$ which states that there are no more than $m_k$ eigenvalues in any subinterval $J \subset I$ with $|J| \leq L^{-q}$ since we can find $q > 2d$ so that $\beta - \frac{q}{2} > q$.
\end{proof}

Further investigations on the simplicity of eigenvalues for Anderson-type models may be found in the article by Naboko, Nichols, and Stolz \cite{naboko-nichols-stolz}, Mallick \cite{mallick1}, Mallick and Krishna \cite{mallick-krishna1}, and Mallick and Narayanan \cite{mallick-narayanan1}. Mallick \cite{mallick1} proves that the singular spectrum is simple for a class of Anderson models with higher rank perturbation extending the results of Naboko, Nichols, and Stolz \cite{naboko-nichols-stolz}.
Mallick and Krishna \cite{mallick-krishna1} prove that, for  higher rank Anderson models with the single site potential having support in the whole real line, the Minami estimate implies simplicity of the pure point spectrum away from the continuous spectrum. They a also show that in the case of higher multiplicity spectrum the spectral statistics cannot be Poisson but must be compound Poisson. Mallick and Narayanan \cite{mallick-narayanan1} prove that higher rank models on some graphs have eigenvalues of higher multiplicity.


\section{Decorrelation estimates for the discrete alloy-type model}\label{sec:alloy-type1}
\setcounter{equation}{0}

In this section, we prove decorrelation estimates for the nonsign definite alloy model studied by Tautenhahn and Veseli\`c \cite{tautenhahn-veselic1}. As above, these imply the asymptotic independence of local eigenvlaue statistics associated with two energies in the localization regime sufficiently far apart. The discrete random Schr\"odinger operator
acting on $\ell^2(\Z^d)$ is described by
\begin{equation}\label{eq:alloy-type-model1}
H_\omega=\mathcal{L} +   V_\omega, 
\end{equation}
where $\mathcal{L}$ is the finite-difference Laplacian, and the random potential $V_\omega$ is defined by
\beq\label{eq:alloy-type-pot1}
V_\omega (m) := \sum_{n \in \Z^d} \omega_n a_{m-n}.
\eeq
The potentials at two sites, $V_\omega (m)$ and $V_\omega (n)$, are independent only if $\| n-m\| > {\rm diam} ~a$.
 Furthermore, the rank of $V_\omega (m)$ is $| \supp ~a|$. The single-site potential $a$ and random variables $(\omega_k)$ satisfy the following hypotheses.


\vspace{.1in}
\noindent
{\bf Hypothesis 1.}
\emph{The single-site potential $a$ is a real, compactly supported function $a: \Z^d \rightarrow \R$
with $a_0>0$ satisfying the condition
\beq\label{eq:alloy-type-ssite1}
 0 < \sum_{n \in \Z^d \backslash \{0 \}} ~ | a_n| \leq a_0 .
\eeq
}

\vspace{.1in}
\noindent
Given a single-site potential $a$, we define a parameter $\delta \geq 0$ by:
\beq\label{eq:delta1}
\delta:= \dfrac{\underset{m\neq 0}{\sum}|a_m|}{a_0}<1,
\eeq

\vspace{.1in}
\noindent
{\bf Hypothesis 2.}
\emph{The single-site potential $a$ is such that the parameter $\delta > 0$.}

\vspace{.1in}
\noindent
The Fourier transform $\hat{a}: \T^d = [0, 2\pi)^d \rightarrow \C$, is defined by
$$
\hat{a}(\theta) := \sum_{k \in \Z^d} ~e^{i \theta \cdot k} a_k, ~~~ \theta \in \T^d,
$$

\vspace{.1in}
\noindent
{\bf Hypothesis 3.}
\emph{The Fourier transform $\hat{a}$ of the single-site potential $a$ is never zero: $\hat{a}(\theta) \neq 0$, for all $\theta \in \T^d$.}

\vspace{.1in}
\noindent
{\bf Hypothesis 4.}
\emph{The family of random variables $(\omega_m)$ are iid random variables with a common, compactly supported density $\rho \in W^{2,1}(\R)$ with support $\rho \subset [-M, M]$, for some $0 < M < \infty$.}

We note that the usual rank one Anderson model corresponds to $a_m = a_0 \delta_{m0}$ so $a$ is supported at a single point and $\delta = 0$. In the case considered here, we will always assume that $\delta > 0$ and the single site potential $a$ has compact support. In particular, there is no restriction on the sign of the terms $a_m$.


Let us write
\beq\label{eq:alloy-cnst1}
m := \sum_n a_n\geq a_0(1-\delta)> 0 .
\eeq
It follows from standard methods that the almost sure spectrum of $H_\omega$ is equal to $[-2d,2d]+ m \cdot\text{supp }\omega_0$. In particular, the almost sure spectrum is a union of intervals and contains at least two intervals $I_1,I_2$ such that $\text{dist}(I_1,I_2) > 4d+mc$, for some $0<c\leq 2M$, only depending on $\supp \omega_0$. We always assume that the constant $(M, c, \delta)$ satisfy the condition
\begin{equation}\label{eq:costant-constraint1}
cM^{-1}(1-\delta)^2-2\delta(1+\delta)>0 .
\end{equation}
Under this condition, we extract from \cite[Corollary 3.4]{tautenhahn-veselic1} the following Minami estimate
(M): There exists $C>0$ such that for all interval $I \in \R$, we have
\begin{equation}\label{eq:alloy-type-minami1}
\P\left({\rm Tr} ~ X_\ell(I) \geq 2\right)\leq C |I|^2\ell^2
\end{equation}

Although not explicitly stated in \cite{tautenhahn-veselic1}, the Minami estimate \eqref{eq:alloy-type-minami1}
and the method of Klein-Molchanov \cite{klein-molchanov}, presented in section \ref{sec:multiplicity1},
allow us to prove that the eigenvalues of the alloy model \eqref{eq:alloy-type-model1}--\eqref{eq:alloy-type-pot1} are almost surely simple. So although the rank of $a$ is greater than one, the standard Minami estimate holds implying simplicity of the eigenvalues in the localization regime and Poisson statistics.

We now turn to the proof of the decorrelation estimates, Proposition \ref{prop:decorrelation-prop1}, for the random alloy model assuming \eqref{eq:costant-constraint1}. Because of the Minami estimate \eqref{eq:alloy-type-minami1}, we may take $m_k = 1$.

We take $E,E' \in \Sigma_{\rm CL}$ and such that $|E-E'|>4d$. We may restrict ourselves to those configurations  $\omega$ such there is one eigenvalue in $I_L(E)$ and one in $J_L(E')$ and such that the distance with the rest of the spectrum of $H_{\omega,\ell}$ is greater than $(L \log L)^{-d}$. By the Wegner estimate, this is possible with probability greater than $1 - (\ell /(L \log L))^{d}$. Let us write $E_j^\ell(\omega)$ and $E_k^\ell(\omega)$ these two eigenvalues with normalized eigenvectors $u_j^\ell$ and $u_k^\ell$. We note that $u_j^\ell(m), u_k^\ell (m) =0$, if $m\notin \Lambda_\ell$.
The results of section \ref{subsec:weight-trace1} hold with $k_1 = 1$, $k_2 = 1$, and $\mathcal{T}_\ell (E, 1,\omega)= E_j^\ell(\omega)$ and
$\mathcal{T}_\ell (E^\prime, 1,\omega)= E_k^\ell(\omega)$.

The first main difference appears in the variational formulas of section \ref{subsec:variation1}, in particular, the lower bound \eqref{eq:variation5}.
In the alloy case, the gradients of the eigenvalues are not normalized.
We prove the following lower bound:

\begin{lemma}\label{lemma:alloy-gradient1}
There exists a finite constant $K>0$, depending only on $M = \sup \omega_0$, and $\delta$ defined in \eqref{eq:delta1}, such that
\begin{equation*}
\left\|\dfrac{\nabla_\omega E_j^\ell(\omega)}{\|\nabla_\omega E_j^\ell (\omega) \|_1}-\dfrac{\nabla_\omega E_k^\ell (\omega)}{\|\nabla_\omega E_k^\ell(\omega) \|_1} \right\|_1\geq K
\end{equation*}
\end{lemma}

%
%
%
\begin{proof}
By the Feynman-Hellmann formula we have
\beq
\partial_{\omega_n} E_j^\ell (\omega) = \sum_{m\in\Z^d} a_{m} |u_j^\ell(m+n)|^2
\eeq
from which it follows that
\begin{equation}
\left|\partial_{\omega_n} E_j^\ell (\omega) - a_0 |u_j^\ell(n)|^2\right|\leq \sum_{m\neq 0} |a_m||u_j^\ell(n+m)|^2 .
\end{equation}
This implies that the $L^1$-norm of the gradient of $E_j^\ell(\omega)$ satisfies
\begin{equation}
\left|\| \nabla_\omega E_j^\ell (\omega)\|_1\ - a_0\right|  \leq \sum_{n\in\Z^d}\sum_{m\neq 0} |a_m||u_j^\ell(n+m)|^2\leq \sum_{m\neq 0}|a_m|.
\end{equation}
Therefore, one has
\begin{equation}
a_0(1-\delta)\leq \| \nabla_\omega E_{j}^\ell (\omega)\|_1 ,  \| \nabla_\omega E_{k}^\ell (\omega)\|_1 \leq a_0(1+\delta),
\end{equation}
and
\begin{equation}
\big|\|\nabla_\omega E_j^\ell (\omega)\|_1-\|\nabla_\omega E_k^\ell (\omega) \|_1\big|\leq 2\sum_{m\neq 0} |a_m|\leq 2\delta a_0.
\end{equation}
It also follows from the Feynman-Hellmann formula that
\begin{align*}
\omega\cdot(\nabla_\omega E_j^\ell (\omega)-\nabla_\omega E_k^\ell (\omega))&=([\Delta- E_j^\ell (\omega)] u_j,u_j)-([\Delta - E_k^\ell (\omega)] u_k,u_k)\\
&=(\Delta u_j,u_j)-(\Delta u_k,u_k) + (E_k^\ell (\omega)-E_j^\ell (\omega)),
\end{align*}
so that
\begin{align*}
M \|\nabla_\omega E_j^\ell (\omega) - \nabla_\omega E_k^\ell (\omega) \|_1 & \geq |\omega\cdot(\nabla_\omega E_j^\ell (\omega) - \nabla_\omega E_k^\ell (\omega))| \geq |E-E'| - 4d > mc ,
\end{align*}
where $m > 0$ is defined in \eqref{eq:alloy-cnst1}.
We can now finally estimate
%
\bea\label{eq:alloy-grad11}
\lefteqn{ \dfrac{\nabla_\omega E_j^\ell (\omega) }{ \|\nabla_\omega E_j^\ell (\omega) \|_1} - \dfrac{\nabla_\omega E_k^\ell (\omega)}{\|\nabla_\omega E_k^\ell (\omega) \|_1} } \nonumber \\
 & = & \dfrac{\|\nabla_\omega E_k^\ell (\omega)\|_1 \nabla_\omega E_j^\ell (\omega) - \| \nabla_\omega E_j^\ell (\omega) \|_1 \nabla_\omega E_k^\ell (\omega)}{\|\nabla_\omega E_j^\ell (\omega) \|_1\|\nabla_\omega E_k^\ell (\omega) \|_1} \nonumber \\
 & = & \dfrac{\|\nabla_\omega E_k^\ell (\omega)\|_1\left[\nabla_\omega E_j^\ell (\omega) - \nabla_\omega E_k^\ell (\omega)\right]+\left[\|\nabla_\omega E_k^\ell (\omega)\|_1-\|\nabla_\omega E_j\|_1\right]\nabla_\omega E_k^\ell (\omega)}{\|\nabla_\omega E_j^\ell (\omega) \|_1\|\nabla_\omega E_k^\ell (\omega)\|_1}, \nonumber \\
 & &
\eea
%
so that
%
\bea\label{eq:alloy-grad-lb1}
\lefteqn{ \left\| \dfrac{\nabla_\omega E_j^\ell (\omega)}{\|\nabla_\omega E_j^\ell (\omega)\|_1} - \dfrac{\nabla_\omega E_k^\ell (\omega)}{\|\nabla_\omega E_k^\ell (\omega)\|_1} \right\|_1}  \nonumber \\
 & \geq & \dfrac{\|\nabla_\omega E_j^\ell (\omega) - \nabla_\omega E_k^\ell (\omega)\|_1 \|\nabla_\omega E_k^\ell (\omega)\|_1-\big|\|\nabla_\omega E_k^\ell (\omega)\|_1-\|\nabla_\omega E_j^\ell (\omega) \|_1\big|\|\nabla_\omega E_k^\ell (\omega)\|_1 }{\|\nabla_\omega E_j^\ell (\omega) \|_1 \|\nabla_\omega E_k^\ell (\omega)\|_1} \nonumber \\
&\geq & \dfrac{mcM^{-1}(1-\delta)-2\delta(1+\delta)a_0}{(1+\delta)^2} \nonumber \\
&\geq & a_0\dfrac{cM^{-1}(1-\delta)^2-2\delta(1+\delta)}{(1+\delta)^2}>0 ,
\eea
giving an explicit formula for the constant $K>0$ in the lemma.
\end{proof}

We also compute
\begin{equation}
\sum_{n\in\Z^d} \partial_{\omega_n} E_{j}^\ell (\omega) = \sum_{n\in\Z^d} \partial_{\omega_n} E_{k}^\ell (\omega) = m>0 ,
\end{equation}
for the constant $m > 0$ defined in \eqref{eq:alloy-cnst1}, and
\begin{equation}
\sum_{n\in\Z^d} \dfrac{\partial_{\omega_n} E_j^\ell (\omega)}{\|\nabla_{\omega_n} E_j^\ell (\omega)\|_1}+\dfrac{\partial_{\omega_n} E_k^\ell (\omega)}{\|\nabla_\omega E_k^\ell (\omega)\|_1}\geq \dfrac{2m}{a_0(1-\delta)}\geq 2.
\end{equation}
Therefore, it follows that
\begin{equation}
\left\|\dfrac{\nabla_\omega E_j^\ell (\omega)}{\|\nabla_\omega E_j^\ell (\omega)\|_1}+\dfrac{\nabla_\omega E_k^\ell (\omega)}{\|\nabla_\omega E_k^\ell (\omega)\|_1} \right\|_1\geq 2 .
\end{equation}

To complete the proof of the decorrelation estimate \eqref{eq:prob1}, we note that the reduction of section \ref{sec:proof-prop1} holds for the alloy-type model. It remains for us to prove the analog of Proposition \ref{prop:vanishing-prob1} for the alloy-type model.

\begin{pr}\label{lemma:alloy-vanishing-prob1}
Let $E, E^\prime \in \Sigma_{\rm CL}$ be two distinct energies with $|E - E^\prime | > 4d$. For any bounded intervals $I,J \subset \R$,
we define $I_L(E) := L^{-d} I + E$ and $J_L(E^\prime) := L^{-d} J + E^\prime$, as above. We write  $X_\ell (I_L(E) := {\rm Tr} E_{H_{\omega, \ell}} (I_L(E))$, and similarly $X_\ell (J_L(E^\prime)$. Then, for any $\beta > \frac{5}{2}$, there exists a scale $L_0 > 0$, so that for any
 $L > L_0$, there exists a constant $C_0> 0$ so that
\beq\label{eq:alloy-vanishing-prob1}
\P \{ (X_\ell(I_L(E))  = 1 ) \cap (X_\ell(J_L(E^\prime)) = 1  ) \} \leq C_0 \frac{K \| \rho \|_\infty^2}{\Delta E - 4d}  \frac{ (C \log L)^{(2+\beta) d}}{L^{2 d}} .
\eeq
\end{pr}

Lemma \ref{lemma:alloy-gradient1} allows us to write the analog of
\bea\label{eq:-alloy-jac-lb2}
\max_{i \neq j \in \Lambda_\ell} J_{ij} ( E_j^\ell (\omega), {E}_k^\ell (\omega))^2  &\geq & \left( \frac{2^3}{ \ell^{5d}} \right)
 \left\| \frac{\nabla_\omega E_j^\ell (\omega)}{\|\nabla_\omega E_j^\ell (\omega)\|_1} -  \frac{\nabla_\omega E_k^\ell (\omega)}{\|\nabla_\omega E_k^\ell (\omega)\|_1 } \right\|_1^2 \nonumber \\
 & \geq &
K^2 \left( \frac{2^3}{ \ell^{5d}} \right) ,
\eea
where $K >0$ is the constant defined in \eqref{eq:alloy-grad-lb1}. Consequently, an estimate of the form \eqref{eq:gradients1} holds for the alloy-type model, and the probability that the normalized gradients are collinear is zero.

With regard to Lemma \ref{lemma:diffeom-new1}, we mention that because of the support of the single-site function $a$, the determinant $f_E(\omega_j, \omega_k) = \det( H_{\omega,\ell} - E)$ is a polynomial of degree $| \supp ~a|$ in each random variable $\omega_j$ and $\omega_k$.
Hence, the Harnack Curve Theorem states that the number of connected components in the zero set of $f_E$ is bounded above by $| \supp ~a|^2$. By the argument in the proof of Lemma \ref{lemma:diffeom-new1}, the number of points in $\varphi^{-1} (p)$, for any $p \in K$, is bounded above by $| \supp ~a|^2$. As this number is independent of $L$, the proof concludes as in section \ref{subsec:joint-prob1}.

\end{document}